\documentclass[10pt,twocolumn,twoside]{IEEEtran}

\usepackage{%
	algorithm,%
	algorithmic,%
	amsmath,%
	amssymb,%
	amsthm,%
	array,%
	cite,%
	multirow,%
	url,%
	xcolor,%
}
\usepackage[caption=false,font=footnotesize]{subfig}
\usepackage{graphicx}

\allowdisplaybreaks[4]

\graphicspath{{./figures/}}

\newcommand{\E}{\mathbb{E}}
\newcommand{\N}{\mathbb{N}}
\newcommand{\R}{\mathbb{R}}

\newcommand{\cM}{\mathcal{M}}
\newcommand{\cP}{\mathcal{P}}
\newcommand{\cR}{\mathcal{R}}

\newcommand{\cT}{\mathcal{T}}
\newcommand{\cX}{\mathcal{X}}

\newtheorem{theorem}{Theorem}
\newtheorem{lemma}[theorem]{Lemma}
\newtheorem{proposition}[theorem]{Proposition}
\newtheorem{corollary}[theorem]{Corollary}

\newcommand{\eq}[1]{\begin{align*}#1\end{align*}}
\newcommand{\eqn}[1]{\begin{align}#1\end{align}}
\newcommand{\EQ}[1]{\begin{equation*}#1\end{equation*}}
\newcommand{\EQN}[1]{\begin{equation}#1\end{equation}}
\newcommand{\meq}[2]{\begin{xalignat*}{#1}#2\end{xalignat*}}

\newcommand{\ieeeproof}[1]{\begin{IEEEproof}#1\end{IEEEproof}}
\newcommand{\set}[1]{\left\{#1\right\}}
\newcommand{\indicator}[1]{1_{\set{#1}}}
\newcommand{\abs}[1]{\lvert#1\rvert}

\newcommand{\opi}{\overline{\pi}}
\newcommand{\upi}{\underline{\pi}}
\newcommand{\oy}{\overline{y}}
\newcommand{\uy}{\underline{y}}
\newcommand{\oj}{\overline{j}}
\newcommand{\uj}{\underline{j}}
\newcommand{\oJ}{\overline{J}}
\newcommand{\uJ}{\underline{J}}
\renewcommand{\ge}{\geqslant}
\renewcommand{\le}{\leqslant}

\newcommand{\nn}{\nonumber}

\newcommand{\heading}[1]{ {\it {\bf #1 } }}

\begin{document}
\title{Mode-Suppression: A Simple, Stable and Scalable Chunk-Sharing Algorithm for P2P Networks}

\author{Vamseedhar Reddyvari, Parimal Parag,~\IEEEmembership{Member~IEEE,} and Srinivas Shakkottai,~\IEEEmembership{Senior Member,~IEEE}
\thanks{V. Reddyvari and S.~Shakkottai are with the Department of Electrical and Computer Engineering, Texas A\&M University, College Station, TX, 77843 USA (e-mail: vamseedhar.reddyvaru, sshakkot@tamu.edu).}
\thanks{P.~Parag is with the Department of Electrical Communication Engineering, Indian Institute of Science, Bengaluru, India (email: parimal@iisc.ac.in).}
}


\maketitle

\begin{abstract}
The ability of a P2P network to scale its throughput up in proportion to the arrival rate of peers has recently been shown to be crucially dependent on the chunk sharing policy employed. Some policies can result in low frequencies of a particular chunk, known as the missing chunk syndrome, which can dramatically reduce throughput and lead to instability of the system. For instance, commonly used policies that nominally ``boost''  the sharing of infrequent chunks such as the well-known rarest-first algorithm have been shown to be unstable. Recent efforts have largely focused on the careful design of boosting policies to mitigate this issue.  We take a complementary viewpoint, and instead consider a policy that simply prevents the sharing of the most frequent chunk(s). Following terminology from statistics wherein the most frequent value in a data set is called the mode, we refer to this policy as mode-suppression.  We also consider a more general version that suppresses the mode only if the mode frequency is larger than the lowest frequency by a fixed threshold.  We prove the stability of mode-suppression using Lyapunov techniques, and use a Kingman bound argument to show that the total download time does not increase with peer arrival rate.    We then design versions of mode-suppression that sample a small number of peers at each time, and construct noisy mode estimates by aggregating these samples over time.  We show numerically that the variants of mode-suppression yield near-optimal download times, and outperform all other recently proposed chunk sharing algorithms. 
\end{abstract}

\IEEEpeerreviewmaketitle

\section{Introduction}\label{intro}

Peer-to-Peer (P2P) file sharing networks such as BitTorrent~\cite{bittorrent} have been studied intensely in recent years, using analytical models, simulation studies, and large scale field experiments.  This interest partly stems from the dominance of P2P as a source of Internet traffic in past years.  Even today, although the traffic fraction has reduced to around 3-4\% in North America, P2P sharing still occupies a significant fraction of about 30\% of traffic in the Asia-Pacific region~\cite{sandvine2016}.  Interest also stems from a desire to understand the thought-provoking phenomenon of apparent scaling up of the throughput of a P2P network as the number of peers grows, which enables them to effectively distribute content with low file-download times during high demand situations called \emph{flash-crowds}.

In a P2P network, a file is divided into fixed-size chunks, and a peer possessing a set of chunks can upload those chunks to other peers that need them.  Once a peer has downloaded all chunks, it could continue to serve other peers or leave the system.  A so-called \emph{seed server} that possesses all chunks and never leaves is often used to ensure that no particular chunk ever goes missing.  It is the feature of integrating the upload capacity of each peer into the system that is supposed to enable system-wide throughput scaling up with the number of peers.  However, since peers can only share chunks that they possess, it is crucial to ensure the wide availability of all chunks to enable maximum usage of available upload capacity with each peer.

The problem of ensuring that all chunks are easily obtainable---ideally by engendering equal numbers of copies of each chunk over the network---was considered by the original designers of P2P networks.
For example, BitTorrent, which is the most popular P2P network protocol, uses an algorithm called \emph{rarest-first} (RF) to try to achieve this goal~\cite{bittorrent}.
Here, the idea is to keep a running estimate of the frequency of all chunks in the system.
When a peer has a chance to download a chunk, it chooses the least frequent (i.e., the ``rarest'') among all the chunks that it needs.
In practice, peers keep track of the frequency of chunks in local subsets.
Intuition suggests that such ``boosting'' of rare chunks might ensure a near-uniform empirical distribution of chunks.

Recent work has postulated that under some conditions, the rarest-first policy used by BitTorrent actually does not achieve its goal, and can actually be harmful to system performance.    In particular,~\cite{hajek2010} studied a chunk-level model of P2P sharing under which new peers that do not possess any chunks arrive into the system at some rate, contacts between peers happen at random, and at each contact a chunk is transferred to a requesting peer under a given policy.  Peers depart immediately after completing the file download.  The objective was to determine if the system is \emph{stable} under a given policy, i.e., at any time is the number of peers that have not yet received the whole file finite or is it exploding to infinity?   The result was that under several policies including rarest-first and random chunk selection, a particular chunk can become very rare across the network---a phenomenon referred to as the \emph{missing chunk syndrome}.  This causes the creation of a large set of peers that are missing only that one chunk, referred to as the \emph{one club}.  In turn, the seed server must serve the missing chunk to almost all peers (which then depart), which means that the system is unstable unless the upload capacity of the seed server is of the order of the arrival rate of peers into the system.  Thus, the phenomenon largely negates the value of the P2P system.  

More recently, experimental studies have revealed that the missing piece syndrome is an observable phenomenon occurring in BitTorrent networks ~\cite{towsley2017}.
The results show that when the seed server has low or intermittent upload capacity, the throughput of the system saturates as the number of peers grows.
In turn, this causes lengthened stay of peers in the system between arrival and completion,
where an increasingly large number of peers are waiting to obtain the final chunk before leaving.
In other words, designing policies that can ensure stability of a P2P network under a fixed seed server capacity for all peer arrival rates is practically relevant.

\subsection{Related Work on Stable Algorithms}\label{sec:algos}

There has been extensive work on P2P networks, and we refer here only to those directly relating to the scaling properties of a single swarm.
A large system assumption was made in~\cite{srikant2004,gustavo2006,ShaJoh10}, and the evolution of peers and seeds is described using a system of differential equations.  While~\cite{srikant2004,gustavo2006} study the stationary regime and indicate the stability of BitTorrent-like systems for all arrival rates, ~\cite{ShaJoh10} considers the transient regime and studies how much seed server capacity is needed to attain a target sojourn time (the time between the arrival of a peer and its completing the file download).
Results on stability and scaling here require that at least a fixed fraction of the peers' upload capacity can always be utilized---an implicit assumption of chunk availability.
As shown in~\cite{hajek2010}, this assumption need not hold for all chunk selection policies, and a chunk-level model is needed for accurate analysis.

Chunk-level models have considered the missing chunk problem from two angles.
The first method is to explicitly insist that peers that have completed the download should stay in the system as servers for some period of time.
For example, ~\cite{lui2009} presents results on fairness vs. system performance based on how long peers stay after completion.
In a more recent work~\cite{hajek2012}, it was analytically shown that the system is stable as long as peers stay long enough to serve of the order of one additional chunk after completion.
Indeed, in the original BitTorrent implementation this often happened naturally, since most users manually stopped participation at some point after download was completed.
However, current implementations allow for the peer to depart immediately after completion, which can lead to the instability observed in~\cite{towsley2017}.

The second method is to assume that peers would leave immediately after completion, and to design the chunk sharing policy such that the missing chunk syndrome is avoided.  Some algorithms of this nature are ``boosting'' policies that can be thought of as modified versions of rarest-first.  For example, the \emph{rare chunk} (RC) algorithm studied in~\cite{reittu2009,reittu2011,oguz2015} picks three peers at random and chooses a chunk that is available with exactly one of the selected peers (called a ``rare'' chunk).  Also studied in~\cite{oguz2015} is a variant of this algorithm called the \emph{common chunk} (CC) algorithm, which proceeds as in the RC algorithm when the peer has no chunks, then follows a policy of sampling a single peer with random selection among its required chunks until it only needs one more chunk, and then proceeds by sampling three peers and only downloading a chunk if every chunk with it appears at least twice with the sampled peers.  However, although stable, these algorithms appear to have long sojourn times in some settings ~\cite{wagner2017}.

More recent work on chunk sharing policies~\cite{wagner2017} describes an algorithm called \emph{group suppression} (GS), which is based on observations made in~\cite{hajek2010}.
The policy is based on computing the empirical distribution of the states in the system, where a state of a peer is the set of chunks available with that peer.
Peers that belong to the state with highest frequency are not allowed to upload chunks to peers that have fewer chunks than themselves, thus suppressing entry into the highest frequency group.
Although this policy appears to have low mean sojourn times in simulations, it can have high variability.  Also, is complex since it requires the knowledge of the entire empirical state distribution.
Furthermore, the authors are only able to prove stability in a P2P network with exactly 2 chunks, while the stability of the general case is left as a conjecture.

A different model is presented in~\cite{vojnovic2008coupon}, wherein peers arrive into the system already possessing one randomly selected chunk.
This system is stable for many policies (including random chunk selection), but is constrained by the fact that the initial chunk has to be provided by the seed server.  Thus, in this case too the seed's capacity must scale with the arrival rate of peers,  and the system might be unstable otherwise.


\subsection{Main Results}

The nominal objective of Rarest-First is to ensure a uniform chunk distribution across the network, which it actually does not achieve in all cases, causing instability as shown in~\cite{hajek2010}.
Our intuition is that rather than following a policy of boosting low-frequency chunks as rarest-first does, simply preventing the most frequent chunk(s) from being shared would allow less frequent chunks to catch up, and drive the empirical distribution of chunks towards the desired uniform distribution.
Implicitly, this would also remove a small fraction of the upload capacity, keeping peers in the system a little longer, and enabling them to share more copies of rare chunks.

Following this intuition, we propose a policy that we call \emph{mode-suppression} (MS), which is based on terminology used in statistics in which the \emph{mode} is the most frequent value(s) in a data set.
In the basic version of this algorithm, we keep track of the frequency of chunks in the system, and when a peer contacts another peer, it is allowed to download any chunk except the one(s) belonging to the mode.
Any chunk may be downloaded if all chunks are equally frequent (i.e., if all chunks belong to the mode).
We extend this idea to a more general version where we do not insist on always suppressing the mode, but only do so when the highest frequency is greater than the lowest frequency by a fixed threshold.
The policy is simple to implement, since all that is needed is the chunk frequency (which is already a part of BitTorrent).

\begin{table*}[ht]
\vspace{-0.1in}
 \centering
 \caption{Comparison of Chunk Selection Policies}{
\vspace{-0.05in}
 \begin{tabular}{|c|c|c|c|c|c|}
 \hline
\multirow{2}{*}{\heading{Policy}}	& \multirow{2}{*}{\heading{$\mathbf{m=2}$}}	&\multirow{2}{*}{ \heading{$\mathbf{m>2}$}}	                  &\multirow{2}{*}{\heading{Information}}      &\multicolumn{2}{c|}{\heading{ Sojourn time}} \\ \cline{5-6}
& & &  & {\it Choose from 1 Peer } & {\it Choose from 3 Peers } \\ \hline \hline
Random	& Unstable	&  Unstable             & None        & N/A (unstable)                         & N/A (unstable)\\ \hline
Rarest-First (RF)	& Unstable	&  Unstable	      & Chunk Frequency &  N/A (unstable)                        &  N/A (unstable)   \\ \hline
Rare Chunk (RC) & Stable	& Stable                &3 Peers      & Bad                       & Good \\ \hline
Common Chunk (CC) & Stable	& Stable              &3 Peers      & Good & Bad \\ \hline
Group Suppression (GS)	& Stable	& Unknown	      & Complete Distribution& Good                       & {Better}     \\ \hline
 Mode-Suppression (MS T=1)	& Stable	& Stable	      &Chunk Frequency& Good                        & Better \\ \hline
Mode-Suppression (MS T=2m )	& Stable	& Stable	      &Chunk Frequency& Best                        & Best \\ \hline
{Distributed Mode-Suppression} (DMS) & Stable  &Unknown& 3 Peers   & Better                       & Best\\ \hline
{EWMA Mode-Suppression} (MS-EWMA)	& Unknown & Unknown & 1 Peer & Better                        & Best \\ \hline
 \end{tabular}}
 \label{tbl:comparison}
\end{table*}

We consider a model similar to~\cite{hajek2010,oguz2015,wagner2017} in which peers that have no chunks enter the system according to a Poisson process with a certain arrival rate.
There is a seed server that has an independent Poisson clock of a fixed rate, and at each clock tick, it contacts a single peer and uploads a chunk to it following a given policy.
Each peer also has an independent Poisson clock of a fixed rate, and at each clock tick, the peer contacts a randomly selected peer and downloads a chunk from it following the same policy.

We have two main analytical results under this model.  First, we show using a Lyapunov drift analysis that the general version of mode-suppression with any finite frequency difference threshold is \emph{stabilizing under all peer arrival rates} in a system in which the \emph{file is divided into any number of chunks}.  Second, we show using Kingman bound arguments that for the general version of mode-suppression, \emph{the sojourn time does not increase with peer arrival rate.}   Hence, mode-suppression appears to be able to reduce chunk sharing just enough to maintain stability, without negatively affecting the sojourn time in a scaling sense.



We also construct two heuristic variants of the idea that only depend on a smaller set of sample statistics.  The first variant is mode-suppression that samples only one peer at a time and uses the history of interactions to compute a noisy mode based on an exponentially weighted moving average estimate of chunk frequency (MS-EWMA).  The second variant, \emph{distributed mode-suppression} (DMS)  samples 3 peers at a time, and uses a noisy mode constructed from only those samples.  It is straightforward to show that DMS is stabilizing in the case of a system with two chunks following the proof in~\cite{oguz2015}.  However, we primarily study the performance of these heuristic variants via simulations.

We simulate all the algorithms by starting the system in a  corner case where one of the chunks is available only at the seed server, and observe the evolution of the system afterwards.  An additional dimension that we explore is the impact on chunk diversity engendered by being able to pick a chunk from the set possessed across multiple peers, i.e., choice of one chunk from one randomly chosen peer, versus choice of one chunk from the chunk-set of three randomly chosen peers.  We empirically find that MS attains its lowest sojourn time when we set the frequency difference threshold for suppression $T=2m,$ where $m$ is the number of chunks that the file is divided into.    We also find that the variants of MS preformed the best overall, and the case of choosing a chunk from the chunk-set of 3 random peers is near-optimal in terms of sojourn time.   A comparison is presented in Table~\ref{tbl:comparison}.

A preliminary version of this work was presented in \cite{RedPar18}, which only considered the stability of a basic version of mode-suppression.  The current work derives a stability result for a generalized version of mode-suppression that has a frequency difference threshold, empirically determines the right threshold, and develops a Kingman-bound-based sojourn time scaling result.  It thus generalizes and adds to the methodological contributions, as well as to the empirical study.

\section{System Model}
\label{sec:SysMod}
We consider a P2P file sharing system for a single file divided into $m$ chunks.
This file sharing system has a unique seed that has all $m$ chunks,
and the seed stays in the system indefinitely.
Peers arrive according to a Poisson process with rate $\lambda$.
Each incoming peer arrives without any chunks and stays in the system till it obtains all $m$ chunks of the file.
In this model, a peer leaves as soon as it has all $m$ chunks of the file.
The peers can receive the chunks in two ways, either directly from the seed or from other peers.

Whenever the seed or a peer contacts another peer, it is deemed as a contact.
Therefore, each peer and the seed have individual contact processes corresponding to the sequence of contact instants.
Upon contact, the seed or the peer transfers a missing chunk to the contacted peer, according to a  \emph{chunk selection policy}.
When chunk selection policy depends solely on the current state of the system, 
it is called a Markov chunk selection policy.

\subsection{Contact Processes}
The time interval between two contacts are assumed to be random, independent, and identically exponentially distributed,
i.e. all contact processes are assumed to be independent and Poisson.
The Poisson contact rate for the seed is assumed to be $U$,
and each peer is assumed to have a common contact rate of $\mu$.

\subsection{State space}
At any time $t$, the number of peers in the system with a proper subset of chunks $S \subset [m]$ is denoted by $X_S(t) \in \N_0 \triangleq \{0, 1, \dots \}$.
The system at time $t$ can be represented by the state
\eq{
X(t) &= (X_S(t): S \subset [m]).
}
The total number of peers at any time $t$ is denoted by
\eq{
|X(t)| &= \sum_{S \subset [m]}X_S(t).
}
For any Markov chunk selection policy, the continuous time process $(X(t), t \geqslant 0)$ is Markov with countable state space $\cX \triangleq \N_0^{\cP([m])\setminus [m]}$.
The \emph{stability region} is defined as the set of arrival rates $\lambda$,
for which the continuous time Markov chain $X(t)$ is positive recurrent.

\subsection{State transitions}
The generator matrix for the process $X(t)$ is denoted by $Q$.
For this continuous time Markov chain, there can only be a single transition in an infinitesimal time.
We denote the system state as $x \in \cX$ just before any transition,
and let $e_S$ be the unit vector in the direction corresponding to a proper subset $S \subset [m]$.

There are three types of possible transitions.
The first type of state transition is the arrival of a new peer,
that leads to an increase in the number of peers with no chunks.
The corresponding transition rate is denoted by
\EQ{
Q(x, x+e_{\emptyset}) = \lambda.
}
The second and third type of transitions occur when a peer with $S \subset [m]$ chunks receives a chunk $j \notin S$ from the contacting seed/peer.
In both these cases, the next state is denoted by $\cT_{S,j}(x)$.
The second type of state transition occurs when the reception of new chunks doesn't lead to a departure.
This transition is denoted by
\EQ{
\cT_{S,j}(x) \triangleq x -e_S + e_{S \cup \{j\}},~\quad x_S > 0, |S| < m-1.
}
The third type of state transition occurs for a peer with $m-1$ chunks,
which departs the system after getting the last chunk upon contact.
This transition is denoted by
\EQ{
\cT_{S,j}(x) \triangleq x - e_S,~\quad x_S > 0, |S| = m-1.
}
As all possible state transitions fall into one of the above three cases, the rate of transition
for any other pair of states would be $0$.
At a system state $x$, if the contacting source has $B$ chunks and the contacted receiving peer has $S$ chunks,
then the set of available chunks that can be transferred is $B \setminus S$.
Selection of which chunk to transfer is called the \emph{chunk selection policy},
which governs the evolution of the process $X(t)$.  
In particular, the last two transition rates $Q(x, \cT_{S,j}(x))$ can only be computed for a specific Markov chunk selection policy.
We describe the proposed chunk selection policy and the corresponding transition rates in the following section.

\section{Mode-Suppression policy}
\label{sec:SelectPolicy}

In this section, we describe the general version of the
Mode-Suppression (MS) policy (with a finite threshold $T$) and derive its rate transition matrix.
First, we establish some notation.
The set of allowable transfers from a peer with set of chunks $B$ to a peer with set of chunks $S$, is denoted by $A(x, B, S) \subseteq B \setminus S.$
The cardinality of this set is denoted by $h(x,B,S)$, and it takes integer values between $0$ and $m$.
Recall that the seed has all the chunks, and hence the set of allowable chunk transfers by the seed is $A(x, [m], S)$.
Below, we describe the specifics of selecting the set of allowable transfers.

If there are no peers in the system, there is no need for chunk transfer.
Hence, without loss of generality, we consider the mode-suppression policy when there exist peers in the system, or $|x| > 0$.
Here, we assume that each peer has the knowledge of all chunk frequencies in the system.
The frequency of the $j$th chunk is
\begin{equation}
\label{eq:pi_j}
\pi_j(x) \triangleq \frac{\sum_{j \in S}x_S}{|x|}.
\end{equation}
Let $\opi(x)$ and  $\upi(x)$ denote the maximum and minimum chunk frequencies, respectively, in a state $x$.  Then we have
\eq{
&\opi(x) = \max\{\pi_j(x) : j \in [m] \} \quad \text{ and }\\
& \upi(x) = \min\{\pi_j(x): j \in [m] \}.
}

The chunk indices that attain the highest frequency $\arg\max\{\pi_j(x): j \in [m]\}$ are called the modes of the chunk frequencies.
The set of modes is defined as
\EQ{
\cM(x) \triangleq \{j \in [m]: \pi_j(x) = \opi(x)\}.
}
We denote the number of peers with chunk $j$ as
\EQN{
\label{eqn:NumChunks}
y_j(x) \triangleq \sum_{S: j \in S}x_S = \pi_j(x)|x|.
}
We can also define the number of peers with the maximum and the minimum chunk frequency by $\oy(x)$ and $\uy(x)$ respectively.
The number of chunks in the system is denoted by
\EQ{
r(x) \triangleq \sum_{S \subset [m]}|S|x_S = \sum_{S \subset [m]}x_S\sum_{j \in [m]}1_{\{j \in S\}} 
= \sum_{j \in [m]}y_j(x).
}
We can lower bound the total number of chunks by the number of most popular chunk.
Therefore,
\EQN{
\label{eqn:NumChunksLB}
r(x) = \sum_{j \in [m]} y_j(x) \geq \oy(x).
}

For simplicity of presentation, we would drop the dependence on the state $x$ for $y, \pi, r$ when the underlying state $x$ is clear from the context.
Next, we can find a quick bound on the fraction of peers with least popular chunk from its definition.
\begin{lemma}
\label{lemma:pi-bound}
For each state $x$, the fraction of peers with least popular chunk is upper bounded by $\upi \le \frac{m-1}{m}$ and hence $(1-\upi) \ge \frac{1}{m}$.
\end{lemma}
\ieeeproof{
Any peer in the system can have at most $m-1$ pieces,
or else it would leave the system.
The result follows from bounding the total number of pieces in the system as
\EQ{
\sum_{j=1}^m \pi_j |x| = \sum_{S \subset [m]}|S|x_S \le (m-1)\sum_{S \subset [m]}x_S = (m-1)|x|.
}
Since $\pi_j \ge \upi$ for each $j \in [m]$, we have $m \upi |x|  \le (m-1) |x|$.
Hence, it follows that $\upi \le (1 -\frac{1}{m})$ and the result follows.
}

Now, we will describe the mode-suppression policy. The mode-suppression policy {\it restricts transmission of any chunk that belongs to the set of modes if its count is greater than count of the least frequent chunk by at least $T>0$ units}.
Denote the set of suppressed chunks in state $x$ by $D_{T}(x)$ for some threshold $T \in \N $.
According to MS $D_T(x)$ is given by
\EQN{
 \label{eq:D_T}
D_{T}(x) = \big\{k \in \cM(x) : y_k(x) \ge \uy +T \big\}.
}
The allowable transfer set for MS policy is
\begin{align}
A(x,B, S) &= B \backslash (S \cup D_{T}(x)).
\end{align}

The steps of the mode-suppression policy are shown in Algorithm~\ref{algo:MS} for a generic peer $p.$
\begin{algorithm}
  \caption{Mode-Suppression Policy for peer $p$ }
  \label{algo:MS}
\begin{algorithmic}
  \STATE $S \leftarrow $ Chunk profile of $p$  \\
  \WHILE{$S \neq [m]$}
  \STATE $t \leftarrow t+ \tau, \text{ where } \tau \sim \exp(\mu)$   \\
  \STATE $x \leftarrow X(t)$ \\
  \STATE $\forall j \in [m],$ compute $y_j(x)$ from \eqref{eqn:NumChunks} and $D_T(x)$ from \eqref{eq:D_T} \\
  \STATE Pick a source peer ($B$) randomly \\
  \STATE Choose a chunk ($j$) randomly from $B \backslash \Big(S \cup D_T(x)\Big)$ \\
  \STATE Update $S \leftarrow S \cup \{j\}$ \\
  \ENDWHILE
\end{algorithmic}
\end{algorithm}

The policy of the seed will be similar except for two differences. First, since the contact rate is $U$, $\tau \sim \exp(U)$ and second, seed pushes the chunk to the peer instead of pulling. 

\subsection{Properties of the suppressed set}
Note that if we set $T=1$,  the policy strictly suppresses the mode, and as $T$ becomes larger, we increasingly relax suppression.   When $T \to \infty$, MS is equivalent to the Random Chunk selection policy as there will not be any suppression, and chunks are chosen uniformly and at random. When the difference in number of peers with different chunks are all within threshold $T$,
no chunks are suppressed.
In this case, $D_T(x) = \emptyset$, and we can upper bound the fraction of peers with most popular chunk by the following Lemma.
\begin{lemma}
  \label{lemma:pi-bound2}
  If $D_T(x) = \phi$, then if $|x| > 2Tm$, then $\opi \le 1- \frac{1}{2m}$.
\end{lemma}
\ieeeproof{
Since $D_T(x) = \phi$, we have $(\opi - \upi) |x| \le T$,
and Lemma~\ref{lemma:pi-bound} implies that $\upi \le \frac{m-1}{m}$.
Using these results and the hypothesis $|x| > 2Tm$, we have the result
\EQ{
  \opi \leq \frac{T}{|x|} + \upi   \le \frac{m-1}{m} + \frac{T}{2Tm}  =  1 - \frac{1}{2m}.
}
The result implies that $1-\opi \ge \frac{1}{2m}$.
}

{We would like to make two important observations regarding the suppressed set $D_T(x)$. 
We first observe that depending on the threshold $T$, either all modes are suppressed  or none of the modes are suppressed. 
When $\cM(x) = [m]$, then $D_T(x) = \emptyset$ by definition. 
Hence, we consider $\cM(x) \subset [m]$. 
Since $y_k = \bar{y}$ for all $k \in \cM(x)$, we have 
\EQ{
D_T(x) = \begin{cases}
\cM(x), & \text{ if }\oy \ge \uy + T,\\
\emptyset, & \text{ otherwise}.
\end{cases}
}
We next observe that, 
using the definition of chunk frequency, the set of suppressed states can be written as 
\EQ{
D_T(x) = \cM(x) \indicator{\opi \ge \upi +\frac{T}{\abs{x}}} + \emptyset\indicator{\opi < \upi +\frac{T}{\abs{x}}}.
}
That is, the mode-suppression threshold is a function of the peer population. 
As the peer population $\abs{x}$ grows large, the policy strictly suppresses the mode for $\abs{x} \ge T$. 
Contrastingly for small peer population $\abs{x}=1$,  the policy is most relaxed. 
}

\subsection{Transition rates of the contact process}  
{From the superposition of independent Poisson contact processes, 
the rate at which one of the peers with profile $S$ contacts any other peer is also Poisson with the aggregate rate $\mu x_S$. 
The probability of contacting a source peer with profile $B$ among all peers is $\frac{x_B}{|x|}$. 
From the thinning of Poisson process, we get that the Poisson contact process between any recipient peer with profile $S$ and a source peer with profile $B$ has rate $\mu x_S \frac{x_B}{|x|}$. }

{
The contact process between seed and the peers is an independent Poisson process with rate $U$, 
where the seed contacts any peer at random. 
Hence, the Poisson contact process between seed and any peer with profile $S$ occurs at rate $Ux_S\frac{1}{\abs{x}}$. }

{
Since source peer has $B$ chunks, then it can transfer one out of $h(x,B,S)$ available chunks to the destination peer with $S$ chunks.
The transition of type $\cT_{S,j}$ occurs when one of the peers without chunk $j \notin S$ is contacted by seed or contacts a peer with chunks $B$, and receives the chunk $j$ among all the possible choices. 
From the thinning and superposition of independent Poisson processes, we can write for $j \notin S$ and $x_S > 0$
\eq{
&Q(x,\mathcal{T}_{S,j}(x)) = \\
& \begin{cases}
\displaystyle \frac{x_{S}}{|x|} \bigg( \frac{U }{h(x,[m],S)}+ \mu \sum_{B:j \in B } \frac{x_B}{h(x,B,S)} \bigg) & \mbox{ if } j \notin D_{T}(x),\\
\displaystyle 0 & \mbox{ if } j \in D_{T}(x).
\end{cases}
}
All other entries in the rate transition matrix other than the diagonal entries are $0$, and the diagonal entries are equal to the negative sum of rest of the entries in that row.}
 
  
{
It is difficult to work with exact transition rates for all transitions from state $x$ to state $\cT_{S,j}(x)$. 
We can lower bound the system performance by lower bounding the transition rates when $S \subset \set{j}^c$. 
To this end, we look at the the Poisson contact process of either the seed or one of the peers with chunk $j$ with any peer, with the aggregate rate 
\EQ{
R_j(x) \triangleq U + \mu\sum_{B: j \in B}x_B = U + \mu y_j(x)
}
from the superposition of independent Poisson contact processes. 
The following lemma gives us a lower bound on the transition rates when $S \subset \set{j}^c$,
and the exact transition rate when $S = \set{j}^c$, in terms of the rate $R_j = U + \mu y_j$. }
\begin{lemma}
\label{lem:LowerBoundTransitionRate}
The transition rate from state $x$ to state $x' = \cT_{S,j}(x)$ for the mode-suppression peer-to-peer system is lower bounded by
\EQN{
\label{eqn:LowerBoundTransitionRate}
\frac{x_S}{|x|}R_j \ge Q(x,x') \ge \frac{x_S}{m|x|}R_j, \text{ when } S \subset \set{j}^c.
}
When $S = \set{j}^c$, we can write the corresponding transition rate as
\EQN{
\label{eqn:EqualityTransitionRate}
Q(x,x') = \frac{x_S}{|x|}R_j.
}
\end{lemma}
\ieeeproof{
For $S \subset \set{j}^c$,
we can trivially bound the cardinality of the allowable transfers by
\EQ{
1 \le \inf_{j \in B} h(x, B, S) \le \sup_{j \in T} h(x, B, S) \le m.
}
This provides the bounds on the transition rate.

When $S = \{j\}^c$,
it is clear that the set of allowable transfer is $\{j\}$ for the contacting sources.
Hence, $h(x, B, S) = |S^c| = 1$ and the equality for transition rate follows.
}

\section{Stability Region of Mode-Suppression}
In this section we characterize the stability region of mode-suppression.
To prove the positive recurrence of the associated continuous time Markov chain $X(t)$, we employ the Foster-Lyapunov criteria.

{\bf Foster Lyapunov Criteria: } {\it Let $\phi$ be a time homogenous, irreducible and continuous time Markov process and $\cX$ be its state space. If there exists a finite set of states $F\subset {\cX}$, a Lyapunov function $V: \cX \to (0,\infty)$ and some constants $b>0$, $\epsilon>0$, such that
\EQ{
    QV(x) \le  -\epsilon + b \indicator{x \in F} \quad \forall x\in {\cX},
}
then $\phi$ is positive recurrent \cite{Meyer-Tweedie, wagner2017}.}

We consider the following Lyapunov function, 
\EQN{
\label{eqn:LyapunovFunction}
V(x) \triangleq \sum_{i=1}^m(\oy - y_i)^2 +  C_1 (|x|- \oy) +  C_2( M - r)^+,
}
where, $C_1,C_2$ and $M$ are positive constants that satisfies the constraints, $C_1 > (2T-1)(m-1)$, $C_2 \geq \frac{2m^2(C_1 \lambda + \epsilon)}{U}$, $M >  \max \left\{mN_{21},N_{22}\right\},$ where $\epsilon > 0$ and $N_{21},N_{23}$ are positive constants defined in the equations~\eqref{eqn:n21} and~\eqref{eqn:n23}. Note that the explicit dependencies of $\pi(x)$ and $y(x)$ on $x$ are not shown for simplicity.

The intuition behind this Lyapunov function is as follows.  The nominal objective of MS is to approximately attain a uniform distribution of chunks (with the allowable error being related to the threshold value $T$).   Hence, we should expect that the policy should promote negative Lyapunov drift whenever the current state differs from uniformity.  Our Lyapunov function is designed to penalize three cases, namely, (i) where chunks have significantly differing frequency, (ii) where some might have zero frequency, and (iii) where all have zero frequency.

For a Markov process $X(t)$ with associated generator matrix $Q$,
the expected rate of change of potential function from state $x$ is called the mean drift from this state,
and is given by
\EQ{
 QV(x) \triangleq \sum_{y}Q(x,y)(V(y)-V(x)).
}
The mean drift from a state $x$ for the Markov process $X(t)$ for the mode-suppression policy,
in terms of its generator matrix $Q$ can be written as
\eqn{
\label{eq:DeltaV}
QV(x) &= Q(x,x+e_{\emptyset})(V(x+e_{\emptyset})-V(x)) \\
&+ \sum_{j \in [m]} \sum_{S :j \notin S} Q(x,\cT_{S,j}(x))(V(\cT_{S,j}(x))-V(x)). \nonumber
}

First, we compute the mean drift corresponding to a new peer arrival.
The arrival of a new peer does not change the number of peers with chunk $j \in [m]$.
However, it does lead to a unit increase in the number of peers in the system.
That is,
\EQ{
Q(x,x+e_{\emptyset})(V(x+e_{\emptyset})-V(x)) = \lambda C_1.
}
We observe that the set  of chunks $S$ such that $j \notin S$ is identical to $S \subseteq \set{j}^c$.
Hence, we can write the mean drift $QV(x)$ from state $x$ in~\eqref{eq:DeltaV} to be equal to
\EQ{
\lambda C_1 + \sum_{j \in [m]}\bigg(\sum_{S \subset \set{j}^c} + \sum_{S = \set{j}^c}\bigg)Q(x,x')(V(x')-V(x)),
}
where $x' = \cT_{S,j}(x)$.
We have the following lemma upper bounding the difference in Lyapunov function between state $\cT_{S,j}(x)$ and $x$.
\begin{lemma}
\label{lem:UpperBoundLyapunovDiff}
For a fixed state $x$,
we can upper bound the difference between Lyapunov functions for states $x$ and $x' =  \cT_{S,j}(x)$ for $j \notin \cM(x)$ as
\EQ{
V(x') - V(x) \le  - U_{11,j}(x) - U_{21}(x)\indicator{S \subset \set{j}^c} + U_{22}(x)\indicator{S =  \set{j}^c}.
}
The corresponding difference between Lyapunov functions for states $x$ and $x' =  \cT_{S,j}(x)$ for $j \in \cM(x)$ is upper bounded by
\EQ{
V(x') - V(x) \le  - U_{12,j}(x) - U_{21}(x)\indicator{S \subset \set{j}^c} +U_{22}(x)\indicator{S =  \set{j}^c},
}
where the following upper bound terms depend only on state $x$,
\eq{
U_{11,j}(x) &\triangleq 2(\oy(x)-y_j(x)) - 1,\\
U_{12,j}(x) &\triangleq C_1 - \sum_{i \neq j}(1+ 2(\oy-y_i)),\\
U_{21}(x) &\triangleq C_2\indicator{M > r(x)},\\
U_{22}(x) &\triangleq C_2(m-1)\indicator{M+m-1 \ge r(x)}.
}
\end{lemma}
\ieeeproof{
We can write the Lyapunov function defined in~\eqref{eqn:LyapunovFunction}
as sum of two functions, $V(x) = V_1(x) + V_2(x)$,
where
\eq{
V_1(x) &\triangleq \sum_{i=1}^m(\oy(x) - y_i(x))^2 + C_1(\abs{x} -\oy(x)),\\
V_2(x) &\triangleq C_2(M-r(x))_+.
}
The transitions $\cT_{S,j}(x)$ occur for sets $S \subseteq \set{j}^c$.
We will consider the following two cases.

\textbf{Case $S \subset \set{j}^c$.}
In this case, a transition $\cT_{S,j}(x)$ leads to a peer with set $S$ of chunks receiving chunk~$j$.
This keeps the number of peers unchanged and $|x'| = |x|$.
This transition leads to a unit increase in the number of peers with chunk $j$, and no change in the number of peers with other chunks.
That is, $y_i(x') = y_i(x) + \indicator{i = j}$. 
This implies a unit increase in the number of chunks in the system, i.e. $r(x') = r(x) + 1$.
Hence, 
\EQ{
\frac{V_2(x') - V_2(x)}{C_2} = (M-r-1)_+ - (M-r)_+ = -\indicator{M > r}.
}
For $j \notin \cM(x)$,  we have $\oy(x') = \max_{i}y_i(x') = \oy(x)$, and hence
\EQ{
V_1(x') - V_1(x) = (\oy - y_j-1)^2-(\oy- y_j)^2.
}
When $j \in \cM(x)$, we have $\oy(x') = \oy(x)+1$, and hence
\EQ{
V_1(x') - V_1(x) = \sum_{i \neq j}\Big((\oy - y_i+1)^2-(\oy- y_i)^2\Big) -C_1. 
}

\textbf{Case $S = \set{j}^c$.}
In this case, a transition $\cT_{S,j}(x)$ leads to a departure of peer that had chunks $\set{j}^c$.
That is, $|x'| = |x| - 1$.
This leads to no change in the number of peers with chunk~$j$,
and a unit decrease in the number of peers with other chunks.
That is, $y_i(x') = y_i - \indicator{i \neq j}$. 
This implies decrease in number of chunks in the system by $m-1$, i.e. $r(x') = r(x) - m + 1$.
Hence, we have
\eq{
\frac{V_2(x') - V_2(x)}{C_2} &= (M-r+m-1)_+ - (M-r)_+ \\
&\le (m-1)\indicator{M + m -1 \ge r}.
}
For $j \notin \cM(x)$,  we have $\oy(x') = \max_{i}y_i(x') = \oy(x)-1$, and hence
\EQ{
V_1(x') - V_1(x) = (\oy - y_j-1)^2-(\oy- y_j)^2. 
}
When $j \in \cM(x)$, we have $\oy(x') = \oy(x)$, and hence
\EQ{
V_1(x') - V_1(x) = \sum_{i \neq j}\Big((\oy - y_i+1)^2-(\oy- y_i)^2\Big) - C_1. 
}
Result follows from combining both the cases for $j \notin \cM(x)$ and $j \in \cM(x)$.
}
We note that when $D_T(x) \neq \emptyset$,
transition to state $\cT_{S,j}(x)$ is possible only for $j \notin D_T(x) = \cM(x)$.
That is, we have $j \notin \cM(x)$ for any transition to state $\cT_{S,j}(x)$.
When $D_T(x) = \emptyset$,
transition to state $\cT_{S,j}(x)$ is possible for all $j \in [m]$.
In particular, it is possible that $j \in \cM(x)$.

\begin{corollary}
\label{cor:NegativeDiff}
Let $x' = \cT_{S,j}(x)$ for $S \subseteq \set{j}^c$.
We can write the following inequality on the Lyapunov function difference
\EQN{
\label{eqn:PotDiffNonArrivalNoSuppress}
V_1(x')- V_1(x) \le \begin{cases}
-1, & j \notin \cM(x), \\
-(C_1 - (2T-1)(m-1)), & j \in \cM(x).
\end{cases}
}
That is, when $C_1 > (2T-1)(m-1)$, the potential difference $V_1(x')- V_1(x) < 0$ for all $j \in [m]$.
\end{corollary}
\ieeeproof{
From the definition of $\cM(x)$, we have $\oy \ge y_j + 1$ for all $j \notin \cM(x)$.
Hence, we have $1 -2(\oy - y_j) \le -1$ for all $j \notin \cM(x)$.

Next, we consider the case when $D_T(x) = \emptyset$.
In this case, a transition to state $x' = \cT_{S,j}(x)$ is possible for $j \in \cM(x)$.
Further, it implies that
\EQ{
1 + 2(\oy - y_i) \le
\begin{cases}
1, & i \in \cM(x),\\
2T - 1, & i \notin \cM(x).
\end{cases}
}
Therefore, for all $j \in \cM(x)$, we have $\sum_{i \neq j}(1 + 2(\oy - y_i)) \le (2T-1)(m-1)$.
}
\begin{lemma}
\label{lem:FracPeerWOSC}
Let the fraction of peers without single chunk $j$ be denoted by $\gamma_j \triangleq \frac{x_{\set{j}^c}}{\abs{x}}$,
 then $\sum_{S \subset \set{j}^c}x_S = (1 - \pi_j - \gamma_j)\abs{x}$ and $\gamma_j \le \opi$.
\end{lemma}
\ieeeproof{
From the definition of the total number of peers $\abs{x} = \sum_{S \subset [m]}x_S$ and the number of peers $y_j = \sum_{S: j \in S}x_S$ with chunk $j$, we get
\EQ{
|x| - y_j = \sum_{S \subseteq \set{j}^c}x_S = \sum_{S \subset \set{j}^c}x_S +x_{\set{j}^c}.
}
In terms of $\gamma_j = {x_{\set{j}^c}}/{\abs{x}}$ and $\pi_j = y_j/\abs{x}$,
we can write $\sum_{S \subset \set{j}^c}x_S = (1 - \pi_j - \gamma_j)\abs{x}$.
We also observe that
\eqn{
\label{eqn:FracPeer}
x_{\set{j}^c} &\le \sum_{S: i \in S, i \neq j}x_S = y_i\indicator{i \neq j} \le \oy.
}
When there are no peers in the system, i.e. $\abs{x} = 0$, we have $\gamma_j = \opi = 0$.
For $\abs{x} > 0$, dividing both sides of the above equation by the number of peers $\abs{x}$,
we get that $\gamma_j \le \opi$.
}

\begin{proposition}
\label{prop:MeanDriftUB1}
An upper bound on the mean drift from any state $x$ such that $D_T(x) \neq \emptyset$ is
\eq{
QV(x) \le &\lambda C_1- \sum_{j \notin \cM(x)}\frac{R_j}{m}\Big[U_{11,j}(x)(1-\pi_j)\\
&+(1-\pi_j-\gamma_j)U_{21}(x) -\gamma_jmU_{22}(x)\Big].
}
\end{proposition}
\ieeeproof{
When $D_T(x) \neq \emptyset$, then transitions from state $x$ to $x' = \cT_{S,j}(x)$ are possible only for chunks $j \notin \cM(x)$.
Hence,
\eq{
QV(x) =& \lambda C_1 + \sum_{j \notin \cM(x)}\bigg(\sum_{S \subset \set{j}^c} + \sum_{S = \set{j}^c}\bigg)Q(x,x')\\
&\bigg(V_1(x')-V_1(x) + V_2(x') - V_2(x)\bigg).
}
We have $V_1(x') - V_1(x) = - U_{11,j}(x) \le -1$ for all $j \notin \cM(x)$ from Corollary~\ref{cor:NegativeDiff}.
For $S \subset \set{j}^c$,
the difference $V_2(x') - V_2(x) = U_{21}(x) = -C_2\indicator{M > r} < 0$ from Lemma~\ref{lem:UpperBoundLyapunovDiff}
and the transitions rate $Q(x,x')$ is lower bounded by $\frac{R_j}{m}$ from~\eqref{eqn:LowerBoundTransitionRate}.
Hence,
\EQ{
Q(x,x')(V(x') - V(x)) \le -R_j\frac{x_S}{m\abs{x}}(U_{11,j}(x) + U_{21}(x)).
}
For $S = \set{j}^c$, the transition rate $Q(x, x') = R_j$ from~\eqref{eqn:EqualityTransitionRate},
and $V_2(x') - V_2(x) \le U_{22}(x) = C_2(m-1)\indicator{M+m-1 \ge r}$ from Lemma~\ref{lem:UpperBoundLyapunovDiff}.
Therefore,
\EQ{
Q(x,x')(V(x') - V(x)) \le -R_j\frac{x_S}{\abs{x}}(U_{11,j}(x)-U_{22}(x)).
}
Summing the above upper bounds for all $S \subseteq \set{j}^c$,
we get the following upper bound on the mean drift
\eq{
 QV(x) \le &\lambda C_1 - \sum_{j \notin \cM(x)}\frac{R_j}{m}\Big[(U_{11,j}(x) + U_{21}(x))\sum_{S \subset \set{j}^c}\frac{x_S}{\abs{x}}\\
&+m(U_{11,j}(x)-U_{22}(x))\sum_{S = \set{j}^c}\frac{x_S}{\abs{x}}\Big].
}
Substituting $\gamma_j$ defined in Lemma~\ref{lem:FracPeerWOSC} and $\pi_j$ defined in equation~\eqref{eq:pi_j},
in the above upper bound on mean drift,
and using the fact that $(1-\pi_j) + (m-1)\gamma_j \ge (1-\pi_j)$,
we get the result.
}

\begin{proposition}
\label{prop:MeanDriftUB2}
An upper bound on the mean drift from any state $x$ such that $D_T(x) = \emptyset$ is
\eq{
QV(x) &\le \lambda C_1- \sum_{j \in [m]}\frac{R_j}{m}\Big[U_{21}(x)(1-\pi_j-\gamma_j)-m \gamma_jU_{22}(x)\\
&+ (1-\pi_j)\big(U_{11,j}(x)\indicator{j \notin \cM(x)} + U_{12,j}(x)\indicator{j \in \cM(x)}\big)\Big].
}
\end{proposition}
\ieeeproof{
When $D_T(x) = \emptyset$, then transitions from state $x$ to $x' = \cT_{S,j}(x)$ are possibly for all chunks $j \in [m]$.
We have $V_1(x') - V_1(x) = - U_{11,j}(x) \le -1$ for all $j \notin \cM(x)$ from Corollary~\ref{cor:NegativeDiff}.
From similar arguments in the proof of Proposition~\ref{prop:MeanDriftUB1}, we have
\eq{
&\sum_{j \notin \cM(x)}Q(x,x')(V(x')-V(x)) \le \\
&-\sum_{j \notin \cM(x)}\frac{R_j}{m}\Big[U_{11,j}(x)(1 - \pi_j + (m-1)\gamma_j) \\
&+ U_{21}(x)(1-\pi_j-\gamma_j) - m U_{22}(x)\gamma_j\Big].
}

Since $C_1 > (2T-1)(m-1)$, we can similarly get
\eq{
&\sum_{j \in \cM(x)}Q(x,x')(V(x')-V(x)) \le \\
&-\sum_{j \in \cM(x)}\frac{R_j}{m}\Big[U_{12,j}(x)(1 - \pi_j + (m-1)\gamma_j) \\
&+ U_{21}(x)(1-\pi_j-\gamma_j) - m U_{22}(x)\gamma_j\Big].
}
Result follows from summing both the upper bounds.
}

When $\cM(x) \subset [m]$, we have $\opi > \upi$ and we denote the set of least frequent chunks by $\uJ(x) \triangleq \set{j \notin \cM(x): y_j = \uy}$,
and the set of most frequent chunks by $\oJ(x) \triangleq \set{j \notin \cM(x): y_j = \uy}$.
We let $\uj \in \uJ(x)$ be one of the least frequent chunks, and $\oj \in \oJ(x)$ be one of the most frequent chunks.
When $\cM(x) = [m]$, all chunks are equally frequent.
\begin{lemma}
\label{lem:largeN}
Let $K_1 >0, K_2 < 2$ be constants.
For each $\epsilon > 0$ there exists an $N(K_1,K_2,\epsilon) \in \R_+$,
such that if $ \oy  \ge N$,
then for $\cM(x) \subset [m]$, we have
\eq{
C_1 \lambda - K_1 \sum_{j \notin \cM(x)}R_j(1-\pi_j)(2(\oy - y_j)-K_2) \le - \epsilon.
}
\end{lemma}
\ieeeproof{
Lower bounding the summation of positive terms over the non-empty set $[m] \setminus \cM(x)$ by a single term corresponding to the least popular chunk $\uj$,
and lower bounding $1-\upi$ by $\frac{1}{m}$ from Lemma~\ref{lemma:pi-bound},
we can upper bound the LHS of the above equation by
\EQ{
C_1\lambda - \frac{K_1}{m}R_{\uj}(2(\oy - y_{\uj})- K_2).
}
To upper bound the above equation,
we define $\eta$ as the ratio of number of peers with the least and the most popular chunks.
That is, we can write $\uy = \eta \oy$ where $\eta \in \left[0, 1- \frac{1}{\oy} \right]$ since $\uy \le \oy-1$.
Since $R_{\uj} = U + \uy\mu = U + \eta\oy\mu$, we can write
\eq{
&R_{\uj}(2(\oy - y_{\uj})-K_2) = (U+\eta\oy \mu )(2\oy(1 - \eta)-K_2)  \\
&= -K_2U + 2U\oy (1-\eta) - K_2\eta \oy\mu+ 2 \oy^2 \mu \eta  (1-\eta).
}
Let us denote the above quadratic expression in $\eta$ by $g(\eta)$.
We can check that $g''(\eta)   = -4 \oy^2 \mu < 0$.
Hence, the function $g(\eta)$ is strictly concave and quadratic in $\eta$, with a unique maximum.
This function attains minimum at the boundary values of $\eta$,
and we can lower bound $g(\eta)$ as
\eq{
&g(\eta) \geq \min\set{g(\eta): \eta \in [0, 1-  \frac{1}{\oy}]} = g(0) \wedge g(1 - \frac{1}{\oy})\\
&= \left[U (2 \oy -K_2) \wedge (2-K_2)(U + \mu (\oy -1)\right].
}
The result follows since $C_1\lambda - \frac{K_1}{m}g(\eta) \le -\epsilon$ if $\oy \ge N$,
where we can choose $N$ to be
\EQ{
\max \set{\frac{1}{2}\left(\frac{C_1 \lambda +\epsilon}{\frac{K_1}{m}U} +K_2\right) ,\left( \frac{C_1 \lambda +\epsilon}{\frac{K_1}{m} (2-K_2)\mu} - \frac{U}{\mu} +1\right)}.
}
}

\begin{theorem}
\label{thm:stability}
The stability region of Mode-Suppression (MS) is $\lambda > 0$ for any finite $T < \infty$, if $m \ge 2, \mu > 0$ and $U > 0$.
\end{theorem}
\ieeeproof{
To prove the positive recurrence of the continuous time Markov chain $X(t)$, we employ the Foster-Lyapunov criteria~\cite{Meyer-Tweedie}.
We consider the Lyapunov function defined in~\eqref{eqn:LyapunovFunction}.

For any $\delta \in (0,1)$,
we can partition the state space into following three regions,
\meq{2}{
&\cR_1 = \set{\opi \ge  \delta},&& \cR_2 = \set{\opi < \delta, \oy \ge \frac{M}{m}},\\
&\cR_3 = \set{\opi < \delta, \oy < \frac{M}{m}}.&&
}
Let us choose a $\delta$ such that $\delta \leq \min \bigg\{ \left(1+\frac{1}{C_2m(m-1)} \right)^{-1}, \left( 1+ \frac{C_2m(m-1)}{C_1-m+1}\right)^{-1},$ $\left( 1+ \frac{C_2m(m-1)}{C_1-(2T-1)(m-1)}\right)^{-1},  \left( 3+ 2m(m-1)\right)^{-1}.\bigg\}$

For each $i \in [3]$, we can further partition each region $R_i$ into
\eq{
\cR_{i1} &= \set{x \in \cR_i: D_T(x) \neq \emptyset},\\
\cR_{i2} &=  \set{x \in \cR_i: D_T(x) = \emptyset, \cM(x) \subset [m]},\\
\cR_{i3} &= \set{x \in \cR_i: D_T(x) = \emptyset, \cM(x) = [m]}.
}
All these regions have countable number of states,
this is due to the fact that the number of peers without any chunks can be arbitrarily large for any state $x$.
We will prove that in each region $\cR_{ij}$ where $i \in \{1,2,3\}$ and $j \in \{1,2\}$,
the mean drift $QV(x) \le - \epsilon$ for all states $x \in R_{ij} \setminus F_{ij}$ for  some finite set $F_{ij}$ dependent on $\epsilon$.

\textbf {Region $\cR_1$:}
We define the following finite set
\EQ{
F_{1} \triangleq \set{\delta |x| \leq (M+m-1)}.
}
Then, for any state $x \in \cR_1\cap F_{1}^c$, we have the fraction of peers with most popular chunk $\opi \ge \delta$ and the number of peers $\abs{x} > \frac{M+m-1}{\delta}$. 
This implies that the number of peers with most popular chunk $\oy = \opi\abs{x} > (M+m-1)$. 
Since the number of chunks in the system $r(x) \ge \oy$ as shown in~\eqref{eqn:NumChunksLB},
for any $x \in \cR_1 \cap F_{1}^c$, we have $U_{21}(x) = U_{22}(x) = 0$.
\begin{itemize}
\item {\bf Region $\cR_{11}\cup\cR_{12}$:}
We define the finite set $F_{11} \triangleq F_1\cup \set{\delta\abs{x} \le N_{11}}$,
where we choose $N_{11} \triangleq N(\frac{1}{m},1, \epsilon)$ from Lemma~\ref{lem:largeN}.
Then it follows that for any state $x \in (\cR_{11}\cup\cR_{12})\cap F_{11}^c \subseteq \cR_1\cap F_1^c$, the number of peers with most popular chunk $\oy > N_{11}$.
Since the upper bound function $U_{11,j}(x) = (2(\oy(x) -y_j(x))-1)$,
and $U_{12,j}(x) \ge 0$ for $x$ in $R_{12}$,
we can bound the mean drift from states $x \in (\cR_{11}\cup\cR_{12})\cap F_{11}^c$ as
\EQ{
QV(x) \le C_1\lambda - \sum_{j \notin \cM(x)}\frac{R_j}{m}(1-\pi_j)(2(\oy -y_j)-1) \le -\epsilon.
}
\item {\bf Region $\cR_{13}$:}
In this region $\cM(x) = [m]$ and the number of peers with each chunk $j$ is identical and hence $y_j = \oy$.
This implies that $U_{12,j}(x) = C_1 - m+1$ and $R_j = U + \mu \oy$ for each chunk $j \in [m]$.
Since we have chosen $C_1 > (2T-1)(m-1)$, it follows that $U_{12,j}(x) > 0$ for all chunks $j \in [m]$.
From Lemma~\ref{lemma:pi-bound}, we know that $1-\upi \ge \frac{1}{m}$, however for this case $\opi = \upi  = \pi_j$ for each chunk $j \in [m]$ and hence $1- \pi_j \ge \frac{1}{m}$ for each chunk $j$.
We define the following threshold
\eq{
N_{13} &\triangleq \frac{m(C_1\lambda+\epsilon)}{\mu \Big(C_1-m+1 \Big)},
}
to define the finite set of states
\EQ{
F_{13} \triangleq F_{1} \cup\set{\delta|x| \le  N_{13}}.
}
It follows that for any state $x \in \cR_{13}\cap F_{13}^c \subseteq \cR_1 \cap F_1^c$,
we have
\EQ{
QV(x) \le C_1\lambda - \sum_{j \in [m]}\frac{R_j}{m}(1-\pi_j)U_{12,j}(x).
}
Since $U_{12,j}(x) = C_1-m+1$ for each chunk $j$, the fraction of peers $(1-\pi_j) \ge \frac{1}{m}$,
and $R_j = \mu \oy+ U \ge \mu \oy = \mu \opi\abs{x} \ge \mu \delta \abs{x} > \mu N_{13}$,
we can re-write the upper bound on mean drift as
\EQ{
QV(x) < C_1\lambda - \frac{\mu N_{13}(C_1-m+1)}{m} = -\epsilon.
}
\end{itemize}
\textbf {Region $\cR_2$:}
For any $x \in \cR_2$, we can upper bound the fraction of peers with most popular chunk $\opi(x) < \delta$,
and hence we can write for any chunk $j \in [m]$
\EQ{
\frac{1}{1-\pi_j(x)} \le \frac{1}{1-\opi(x)} < \frac{1}{1-\delta}.
}
In addition for any $x \in \cR_2$,  the number of peers with most popular chunk $\oy(x) \ge \frac{M}{m}$,
and we know that the fraction of peers $\gamma_j$ missing single chunk $j$ is upper bounded by the fraction of peers $\opi$ with most popular chunk from Lemma~\ref{lem:FracPeerWOSC}.
Combining the two results, we get $\gamma_j(x) \le \opi(x) < \delta$.
Recall that $U_{21}(x) \ge 0$ and $U_{22}(x) \le C_2(m-1)$,
then we can write the following upper bound
\EQ{
m\gamma_jU_{22}(x) \le \gamma_jC_2\frac{(1-\pi_j)}{(1-\pi_j)}m(m-1) \le (1-\pi_j)\frac{C_2m(m-1)\delta}{1-\delta}.
}
\begin{itemize}
\item {\bf Region $\cR_{21}$:}
From Proposition \ref{prop:MeanDriftUB1},
we can upper bound the mean drift $QV(x)$ from any state $x \in \cR_{21}$ by
\EQ{
\lambda C_1 - \sum_{j\notin \mathcal{M}(x)} \frac{R_j(1-\pi_j)}{m}\Big(2(\oy - y_j)-1 -  \frac{C_2 m(m-1)\delta}{1-\delta} \Big).
}
%
Choosing $\delta < (1+C_2m(m-1))^{-1}$,
we get $\frac{\delta}{1-\delta} C_2m(m-1) < 1$.
Therefore, we can apply Lemma~\ref{lem:largeN} for $K_1 = \frac{1}{m}$ and $K_2 = 1 + \frac{\delta}{1-\delta}C_2 m(m-1)$ for the threshold
\begin{equation}
\label{eqn:n21}
N_{21}  \triangleq N\left(\frac{1}{m},\frac{\delta}{1-\delta} C_2 m(m-1)+1,\epsilon\right).
\end{equation}
Choosing $M \ge m N_{21}$, we see that $\oy(x) \ge \frac{M}{m} \ge N_{21}$ for all $x \in \cR_{21}$,
and hence the mean drift $QV(x) \le -\epsilon$ for all such states $x$.

\item {\bf Region $\cR_{22}$:}
With the choice of threshold $N_{21}$ and $M \ge m N_{21}$,
the mean drift $QV(x) \le -\epsilon$ for all states $x \in \cR_{22}$, if we can show that
\EQ{
\sum_{j \in \cM(x)}\frac{R_j(1-\pi_j)}{m}\Big(U_{12,j}(x) 
- \frac{\gamma_j}{1-\pi_j}
U_{22}(x))\Big) \ge 0.
}

To this end, we recall that $U_{12,j}(x) \ge  C_1 - (2T-1)(m-1) \ge 0$ to write
\eq{
&U_{12,j}(x) - \frac{\gamma_j}{1-\pi_j}U_{22}(x))\\
&\ge C_1-(2T-1)(m-1) -\frac{\delta}{1-\delta}C_2(m-1)
}
We see that the choice of $\delta < (1+ \frac{C_2m(m+1)}{C_1-(2T-1)(m-1)})^{-1}$ gives us the desired result.
\item {\bf Region $\cR_{23}$:}
In this region $\cM(x) = [m]$ and the number of peers with each chunk $j$ is identical and hence $y_j = \oy$.
This implies that $U_{12,j}(x) = C_1 - m+1$  for each chunk $j \in [m]$.
From Proposition \ref{prop:MeanDriftUB2},
we can upper bound the mean drift $QV(x)$ for all states $x \in \cR_{23}$ by
\EQ{
\lambda C_1 - \sum_{j\in [m]} \frac{R_j(1-\pi_j)}{m}\Big( \big(C_1 -m +1\big) - C_2 m(m-1) \frac{\delta}{1-\delta} \Big).
}
We can lower bound the contact rate $R_j(x) = U + \mu \oy(x) \ge \mu \frac{M}{m}$ for each state $x \in \cR_{23}$,
and $1-\pi_j \ge \frac{1}{m}$ from Lemma~\ref{lemma:pi-bound} for each state $x$, to get
\EQ{
QV(x) \le \lambda C_1 - \frac{\mu M}{m^2}\Big( \big(C_1 -m +1\big) - C_2 m(m-1) \frac{\delta}{1-\delta} \Big).
}
Since $\delta < \frac{(C_1 -m+1)}{C_2 m(m-1) + (C_1-m+1)}$, $\Big( \big(C_1 -m +1\big) - C_2 m(m-1) \frac{\delta}{1-\delta} \Big) > 0$. Let us define,
\begin{equation}
N_{23} \triangleq \frac{m^2(C_1\lambda+\epsilon)}{\mu\big( \big(C_1 -m +1\big) - C_2 m(m-1) \frac{\delta}{1-\delta} \big)}.
\label{eqn:n23}
\end{equation}
Choosing $M \ge N_{23}$, we see that the mean drift $QV(x) \le -\epsilon$.
\end{itemize}
\textbf {Region $\cR_3$:}
Since $U_{11,j}(x)$ and $U_{12,j}(x)$ are non-negative for all states $x$ and chunks $j \in [m]$,
we can upper bound the mean drift from any state $x$ as
\eq{
QV(x) \le& C_1\lambda - \frac{R_j(1-\pi_j)}{m}\Big(U_{21}(x) \\
&-\frac{\gamma_j}{1-\pi_j}(U_{21}(x)+ mU_{22}(x))\Big),
}
where $j \notin \cM(x)$ for $\cR_{31}$ and $j \in [m]$ for $x \in \cR_{32}$.

For any $x \in \cR_3$, we can upper bound the fraction of peers with most popular chunk $\opi < \delta$ and lower bound the number of peers with most popular chunk $\oy < \frac{M}{m}$.
From Lemma~\ref{lem:FracPeerWOSC} and the fact that $\pi(x) < \delta$,
it follows that $\gamma_j(x) \le \opi(x) < \delta$ as in Region $\cR_2$.
We can also write the following inequality from the fact that $\pi_j \le \opi < \delta$,
\EQ{
\frac{1}{1-\pi_j(x)} \le \frac{1}{1-\opi} < \frac{1}{1-\delta}.
}
Since the number of chunks in the system $r(x) = \sum_{j=1}^my_j \le m\oy < M$,
and therefore $r(x) < M+m-1$.
It implies that $U_{21}(x) = C_2$ and $U_{22}(x) = C_2(m-1)$.
Therefore for any $x \in \cR_3$, we can write
\eq{
&U_{21}(x) -\frac{\gamma_j}{1-\pi_j}(U_{21}(x)+ mU_{22}(x))\\
&\ge C_2\left(1 - \frac{\delta}{1-\delta}(1+m(m-1))\right).
}
Choosing $\delta \le (3+ 2m(m-1))^{-1}$,
we see that $1 - \frac{\delta}{1-\delta}(1+ m(m-1)) \ge \frac{1}{2}$.
From Lemma~\ref{lemma:pi-bound}, we have $1-\pi_j \ge \frac{1}{m}$ and the contact rate $R_j = U + \mu y_j  \ge U$,
and therefore
\EQ{
QV(x) \le C_1\lambda - \frac{UC_2}{2m^2}.
}
Choosing $C_2 \ge \frac{2m^2(C_1\lambda + \epsilon)}{U}$,
we get that the mean drift $QV(x) \le -\epsilon$ for all $x \in \cR_3$.
}

\section{Scaling of Swarm Size and Sojourn Time}
Our next result is on the scaling properties of MS with respect to the peer arrival rate $\lambda$.
We use the following Kingman moment bound to prove the properties.
\begin{theorem}[Kingman moment bound~\cite{hajek2012}]
 \label{thm:kingman}
Let $X$ be a continuous-time, irreducible Markov process on a countable state space $\mathcal{X}$ with generator matrix $Q$. Suppose $V, f,$ and $g$ are nonnegative functions over the state space $\cX$,
and suppose $QV(x) \le -f(x) + g(x)$ for all $x \in \cX$.
In addition, suppose $X$ is positive recurrent, so that the means, $\bar{f} = \pi f$ and $\bar{g} = \pi g$ are well defined. Then $\bar{f} \le \bar{g}$.
\end{theorem}

We then have the following scaling result.
\begin{theorem}
 Under the Mode-suppression policy, the following statements are true. 
  \begin{enumerate}
    \item {\bf (Scaling of Swarm Size)} The average number of peers in the system $L \le C \lambda,$ where $C$ is a constant.
  \item {\bf (Scaling of Sojourn time)} The average sojourn time of the peers $W,$ is bounded.
  \end{enumerate}
\end{theorem}
\begin{proof}
We make use of the Kingman moment bound with the following Lyapunov function,
\begin{align}
V(x) = \sum_{i=1}^m\big( (\overline{\pi} - \pi_i)|x|\big)^2\; +\;  C_1 \big((1- \overline{\pi})\big)|x|.
\end{align}
This is similar to the Lyapunov function~\eqref{eqn:LyapunovFunction} used in the stability theorem except for the last term.
We can make use of most of the results we derived by substituting $C_2 = 0$.
To compute $QV(x)$, we divide the state space into two regions based on whether the set of suppressed chunks $D_T(x)$ is empty or not.

{\bf Region 1: $D_T(x) \neq \emptyset$, or $D_T(x) = \emptyset$ and $\cM(x) \subset [m]$.}\\
Since $C_2=0$, $U_{21} = U_{22} = 0$.
Therefore, using~\ref{prop:MeanDriftUB1}, we can upper bound $QV(x)$ by
\EQ{
\lambda C_1 - \frac{R_{\uj}}{m}(1-\upi)\big(2 (\oy - \uy ) -1\big).
}
Using $(1-\opi) \geq \frac{1}{m}$ from Lamma~\ref{lemma:pi-bound}, we can upper bound this by
\EQ{
\lambda C_1 - \frac{1}{m^2}\big(U +\mu \uy\big)\big((2  (\oy - \uy ) -1\Big).
}
Let, $\phi(\uy) = \frac{1}{m^2}\big(U +\mu \uy\big)\big((2  (\oy - \uy ) -1\Big)$.
The function $\phi(\uy)$ is concave and quadratic in $\uy$ and $ 0 \le \uy \le \oy -1$.
Hence the minimum of $\phi(\uy)$ lies at one of the extreme points, $\{0,  (\oy -1 )\}$.
Therefore, the above expression can be upper-bounded by,
\EQ{
C_1\lambda  - \min \bigg\{ \frac{U}{m^2} (2 \oy -1), \frac{1}{m^2} \big( U + \mu \oy - \mu \big) \bigg\}.
}
This upper bound can be re-written as
\EQN{
\label{eqn:bound-1}
C_1 \lambda - \oy \Big(\frac{2U}{m^2} \wedge \frac{\mu}{m^2}   \Big) + {k}.
}
where, $k > \max \{ \frac{U}{m^2}, \frac{\mu - U}{m^2} \}$ is a constant independent of $\lambda$.

{\bf Region 2: $D_T(x) = \emptyset$ and $\cM(x) = [m]$.}\\
In this case, $\opi = \upi$ and $\oy = \uy$ and we use the upper bound for $QV(x)$ from equation~\ref{prop:MeanDriftUB2} setting $U_{21}=U{22} = 0$,
\EQ{
\lambda C_1 - \frac{R_{\oj}}{m}(1-\opi)\big(C_1 - (2T-1)(m-1) )\big).
}
Recalling that $(1-\opi) = (1-\upi) \ge \frac{1}{m}$ from Lemma~\ref{lemma:pi-bound},
the fact that $R_{\oj} = R_{\uj} = (U + \mu \uy) \ge \uy$ from non-negativity of $U$,
and since $C_1  > (2T-1)(m-1)$,
we can upper bound the RHS of the above inequality by
\EQ{
\lambda C_1 - \oy\frac{\mu}{m^2}\big(C_1 - (2T-1)(m-1) )\big). \nn
}
Combining both upper bounds, 
we obtain
\EQN{
  QV(x) 
  \le C_1 \lambda -  g(\mu,U,m,T,C_1) \oy + k,
}
where $g(\mu,U,m,T,C_1) = \frac{\mu}{m^2}\min \Big\{ \frac{2U}{\mu} ,1, \big( C_1 - (2T-1)(m-1)\big) \Big\}$ and $k =  \max \{ \frac{U}{m^2}, \frac{\mu - U}{m^2} \}$.

Applying the Kingman bound for $f(x) = C_1\lambda +k$ and $g(x) = g(\mu,U,m,T,C_1) \opi |x|$,
we obtain
\EQ{
\E[\oy(x)] \le  \frac{C_1 \lambda + k}{g(\mu,U,m,T,C_1)}.
}
Since the number of peers $y_i$ with chunk $i$ can be upper bounded by the number of peers $\oy$ with most popular chunk,
and hence the number of chunks $r(x) = \sum_{j \in [m]}y_j(x)$ in the system is smaller than $m\uy$.
Note that $r(x)$ is the number of chunks in the system,
and it exceeds the number of peers with a single chunk.
That is, $r(x) = \sum_{S \subset [m]}\abs{S}x_S \ge \sum_{S: \abs{S} \ge 1}x_S$.
Therefore, it follows from the Kingman bound that
\EQN{
\label{eqn-kingman}
\E[\sum_{S: \abs{S} \ge 1}x_S] \le  \frac{m C_1 \lambda + m\;k}{g(\mu,U,m,T,C_1) }.
}

When a peer enters the system, it has no chunks.
We can view the whole peer swarm as composed of two systems,
with system~$0$ consisting  of peers with no chunks,
and system~$1$ consisting of peers that have one or more chunks.
Peers in the system~$0$ move to the system~$1$ upon obtaining any chunk,
as shown in Figure~\ref{fig:littles-law}.
\begin{figure}[h]
  \centering
    \includegraphics[width=0.4\textwidth]{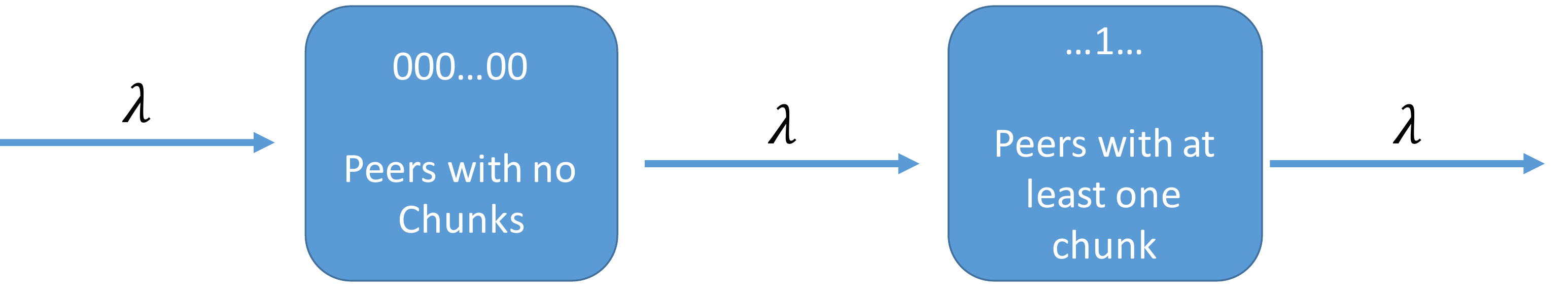}
    \caption{Arrival and departure rates of peers into different systems under MS.  The rates are all $\lambda$ since the system is stable under MS.}
 \label{fig:littles-law}
\end{figure}

Since under MS, the probability of obtaining a chunk is highest when a peer possesses no chunks, the waiting time in system $0,$ denoted by $W_0 $ is upper bounded by the waiting time in  system $1,$ denoted by $W_1$. That is $W_0 \leq W_1$.
Then by Little's Law, the average number of peers in the system
\EQ{
L = \lambda W = L_0 + L_1 = \lambda W_0 + L_1  \le \lambda W_1 + L_1  = 2L_1.
}
Substituting for $L_1$ from the Kingman bound~\eqref{eqn-kingman} we get
\EQ{
L \le \frac{2m (C_1 \lambda + k) }{ g(\mu,U,m,T,C_1)}
\text{ which implies }W \le \frac{2m (C_1  + {\frac{k} {\lambda}}) } { g(\mu,U,m,T,C_1)}.
}
Therefore, the average sojourn time $W$ is bounded and the upper bound on average swarm size $L$ scales linearly with $\lambda$.

\end{proof}

\section{Distributed Policies}\label{sec:distrib}
Although mode-suppression is simple to implement, it does require global information of chunk frequencies.  We now propose two policies that circumvent this requirement.

\subsection{Distributed Mode-Suppression Policy}
Under \emph{distributed mode-suppression} (DMS), a peer contacts three other peers at random, and among the chunks available with more than one peer, we define the \emph{local mode} to be the chunk(s) with greatest frequency.   The peer is allowed to download any chunk that is not part of the local mode. Any chunk may be downloaded if all chunks are equally frequent. 

Let $B^j, j = {1,2,3}$ denote the chunk profiles of three selected peers and $B = \left\{B^1,B^2,B^3\right\}$, 
then we can write the modes 
\eqn{
  \label{eq:M_DMS}
&\cM_{DMS}(x,B) = \nonumber \\
& \left\{ i \in [m] \biggr\rvert \sum_{j=1}^{3}B^j_i \ge \sum_{j=1}^{3}B^j_k, \forall k \in [m],  {\sum_{j=1}^{3}B^j_i > 1 }\right\},
}
and we write the set of suppressed chunks (regardless of whether the chunk is downloaded from the seed or another peer) as 
\EQN{
\label{eq:D_DMS}
D_{DMS}(x,B) = 
\begin{cases}
  \cM_{DMS} 	& \mbox{ if }   \cM_{DMS}\neq [m],\\
  \emptyset 	& \mbox{ if }   \cM_{DMS} = [m].
\end{cases}
}
The steps of the distributed mode-suppression policy are shown in Algorithm~\ref{algo:DMS}.
\begin{algorithm}
  \caption{Distributed Mode Suppression for peer~$p$}
  \label{algo:DMS}
\begin{algorithmic}
  \STATE $S \leftarrow $ Chunk profile of $p$  \\
  \WHILE{$S \neq [m]$}
  \STATE $t \leftarrow t+ \tau, \text{ where } \tau \sim \exp(\mu)$   \\
  \STATE $x \leftarrow X(t)$, $S \leftarrow $ Chunk profile of $p$  \\
  \STATE Select three source peers ($B^i$) randomly \\
  \STATE Compute $D_{DMS}(x,B)$ from~\eqref{eq:D_DMS} \\
  \STATE Choose a chunk~$j$ randomly from $ \cup_{i=1}^{i=3}B_i\backslash \Big( S \cup D_{DMS}(x,S)\Big)$ \\
  \STATE Update $S \leftarrow S \cup \{j\}$ \\
  \ENDWHILE
\end{algorithmic}
\end{algorithm}

\begin{theorem}
The stability region of Distributed Mode-Suppression (DMS) is $\lambda >0$ if $m=2, \mu>0$ and $U >0$.
\end{theorem}
\begin{proof}
The proof for $m=2$ chunks follows using the same Lyapunov function and steps as the
proof of the Rare Chunk policy \cite{oguz2015}, and is hence omitted.
\end{proof}
Stability for the case $m> 2$ chunks is left as a conjecture.

\subsection{EWMA Mode-Suppression} 
Under this policy, each peer calculates the empirical marginal chunk frequencies based only on the chunks possessed by all peers that it has met until (and including) the current time.   The marginal chunk frequency is calculated using an Exponentially Weighted Moving Average (EWMA) taking into account both history and present, and the mode of this estimate is suppressed.

Let $n \in \mathbb{N}$ denote the index of poisson ticks of a peer. We define empirical marginal chunk frequencies of a peer $p$ with $\tilde{\pi}^n(p)$ and are computed as below for each chunk $j \in [m]$,
\begin{align}
  \tilde{\pi}_j^0(p) &= 0, 
 \label{eq:pi_ewma}
\tilde{\pi}_j^n(p) &= (1-\alpha) \tilde{\pi}_j^{n-1}(p) + \alpha B^n_j, 
\end{align}
where $B^n$ denotes the chunk profile of the source peer selected at time slot $n$ by peer~$p$ and $\alpha \in (0,1)$ is the exponential weighting parameter. 
The modes for this policy are defined as 
\begin{align}
\cM^n_{EWMA}(p)& = \left\{ i \rvert \tilde{\pi}^n_i(p) \geq   \tilde{\pi}^n_j(p) \forall j \in [m] \right\},
\end{align}
and the set of suppressed chunks (regardless of whether the chunk is downloaded from the seed or another peer) are denoted by 
\begin{align}
  \label{eq:D_ewma}
D^n_{EWMA}(p) & =\begin{cases}
  \cM^n_{EWMA} & \mbox{ if }   \cM^n_{EWMA}\neq [m],\\
  \emptyset & \mbox{ if }   \cM^n_{EWMA} = [m].
\end{cases}
\end{align}
The steps of EWMA Mode-Suppression Policy is shown in Algorithm~\ref{algo:EWMA}.
\begin{algorithm}
  \caption{EWMA Mode-Suppression for peer~$p$}
  \label{algo:EWMA}
\begin{algorithmic}
  \STATE $S \leftarrow $ Chunk profile of $p$, $n = 0$  \\
  \WHILE{$S \neq [m]$}
  \STATE $t \leftarrow t+ \tau, \text{ where } \tau \sim \exp(\mu)$,  $n \leftarrow n+1 $  \\
  \STATE $x \leftarrow X(t)$  \\
  \STATE Pick a source Peer (B) randomly \\
  \STATE $\forall j \in [m],$ compute $\tilde{\pi}^n_j(p)$ from \eqref{eq:pi_ewma} and $D^n_{EWMA}(p)$ from \eqref{eq:D_ewma} \\
  \STATE Choose a chunk~$j$ randomly from $B \backslash \Big(S \cup D^n_{EWMA}(p)\Big)$ \\
  \STATE Update $S \leftarrow S \cup \{j\}$ \\
  \ENDWHILE
\end{algorithmic}
\end{algorithm}

\section{Simulation Results}
\begin{figure*}[t]
    \includegraphics[width=\textwidth]{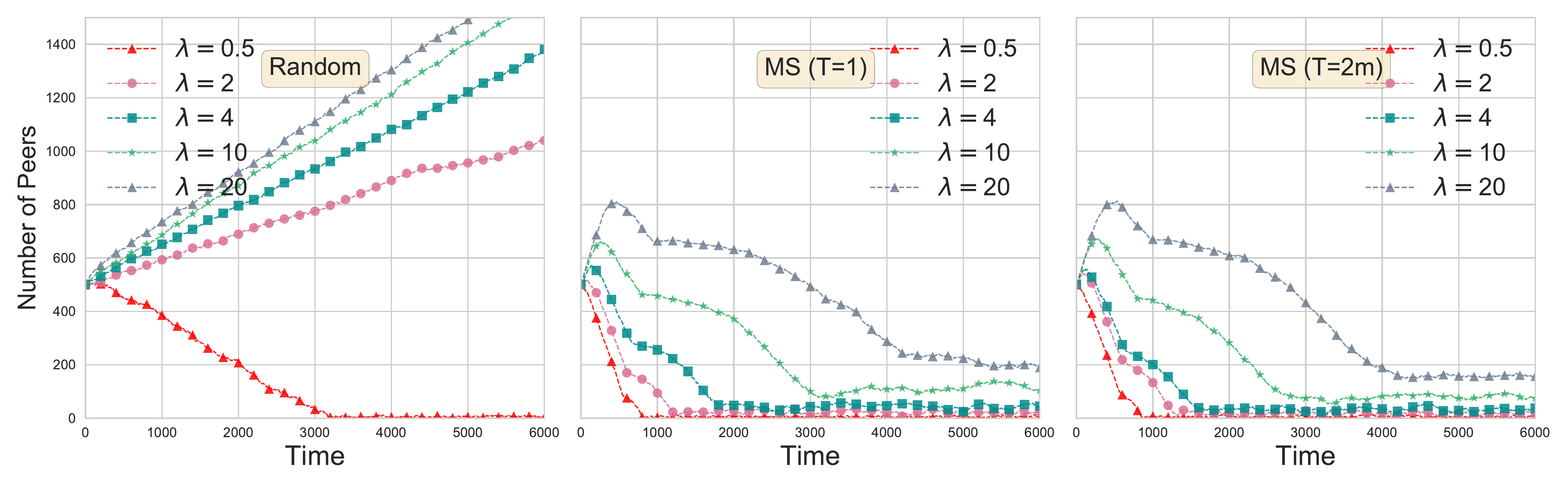}
    \caption{Number of peers in the system when $m=5$, $U=1$ and $\mu=1$.  Random becomes unstable in some cases, whereas MS is always stable.}
  \label{fig:number-of-peers}
\end{figure*}
In this section, we show the results from numerical simulations that illustrate the performance of different chunk selection policies.  Recall that our candidate policies are (i) random chunk selection, (ii) rarest-first, (iii) rare chunk, (iv) common chunk, (v) group suppression,  (vi) mode-suppression, (vii) distributed mode-suppression, and (viii)  mode-suppression-EWMA.  A description of these policies can be found in Sections~\ref{intro}, \ref{sec:SelectPolicy}, and~\ref{sec:distrib}.
For all the simulations, we set the peer contact rate $U$ and seed contact rate $\mu$ as $1$.
To simulate a Poisson process, we make use of the fact that inter arrival times of a Poisson process follow an exponential distribution.
Each peer in the system, including the seed, generates an exponential random variable with mean $\frac{1}{\mu} = \frac{1}{U} = 1$, and the peer or the seed with the smallest value gets a chance to contact another peer.
After the contact, a chunk transfer takes place instantaneously according to the chosen chunk selection policy.

\subsection{Stability of Mode-Suppression Policy}
We begin the simulation with 500 empty peers.  Whenever a peer receives all the chunks, it immediately leaves the system.  In Figure~\ref{fig:number-of-peers}, we plot the number of peers in the system as time progresses for three different polices, namely (i) random chunk selection, (ii) mode-suppression, and (iii) distributed mode-suppression.  The purpose of simulating the random chunk selection policy, which is known to be unstable, is to provide a visual representation of what an unstable regime appears like in order to compare with stable policies.  In this simulation, the number of chunks is taken as $5,$ and the peer arrival rate ($\lambda$) is varied.  We observe that when the peer arrival rate is less than seed rate ($\lambda =0.5 < 1 = U$), the random chunk selection policy is stable.
In all other cases where we have chosen the peer arrival rate $\lambda > U$, the number of peers grows large and the system is unstable.  However, in case of mode-suppression and distributed mode-suppression, the system is stable for all arrival rates.

\subsection{Missing Piece Syndrome in Random Chunk Selection}
We observed in Figure~\ref{fig:number-of-peers} that the random chunk selection policy is not stable when $\lambda > U$.
We illustrate the reason for this instability by observing the evolution of the chunk frequency.
In Figure~\ref{fig:random-chunk}, we plot the time evolution of the number of peers and the fraction of peers having different chunks in the system, for the random chunk selection policy with $m=5$ and $\lambda = 4$.
We see that when number of peers becomes large, one chunk remains rare.
As time progresses, the chunk represented by the red/starred line becomes rare and remains rare forever.
However, all other chunks are available with most of the peers.
This is precisely the formation of the \emph{one-club} caused by the {\it missing piece syndrome}.
\begin{figure}[h]
    \includegraphics[width=0.5\textwidth]{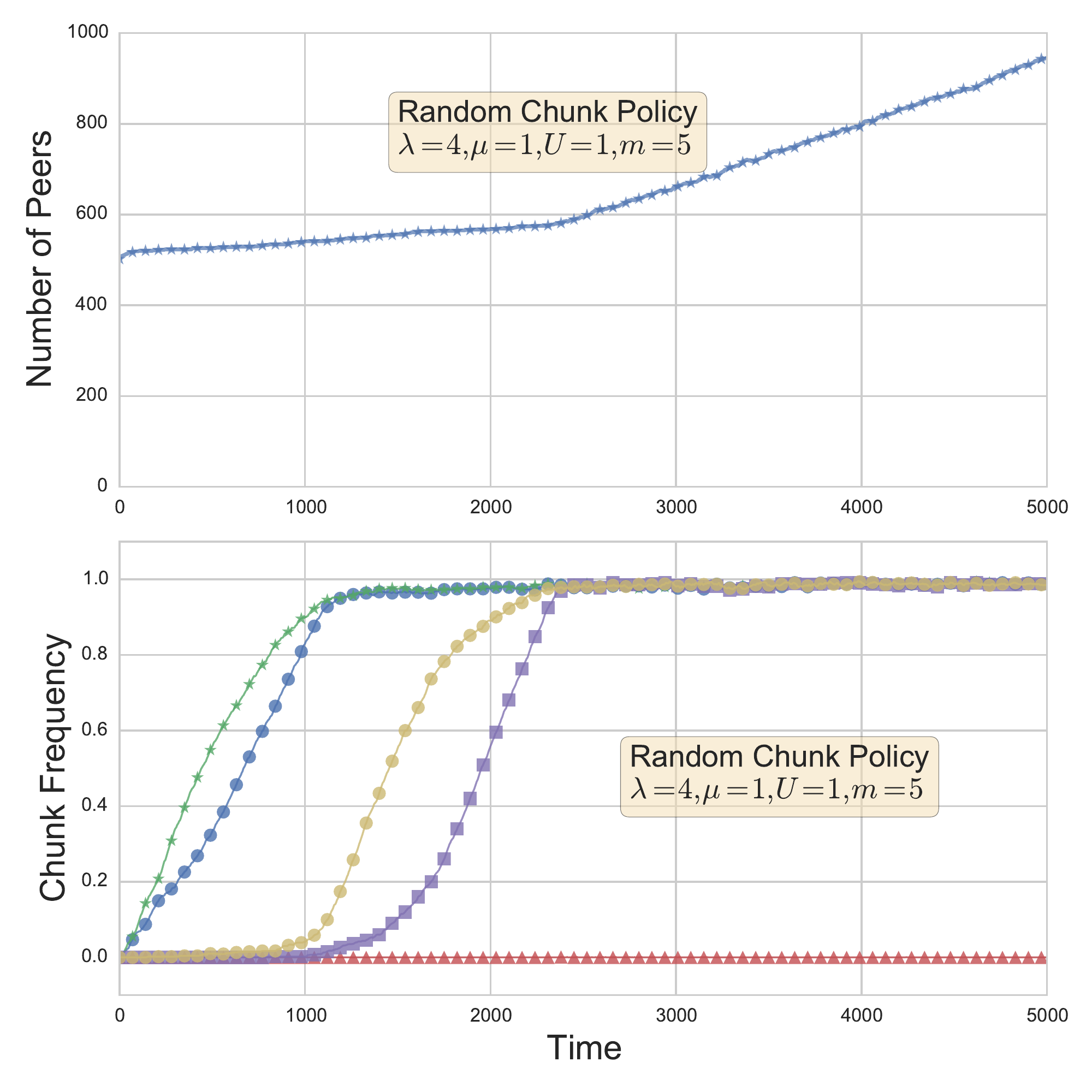}
    \caption{Evolution of peers and chunk frequencies under the random chunk selection policy.  One of the chunks becomes a ``missing chunk'' (red/starred line).}
  \label{fig:random-chunk}
\end{figure}

\begin{figure*}[t]
  \centering
    \includegraphics[width=0.9\textwidth]{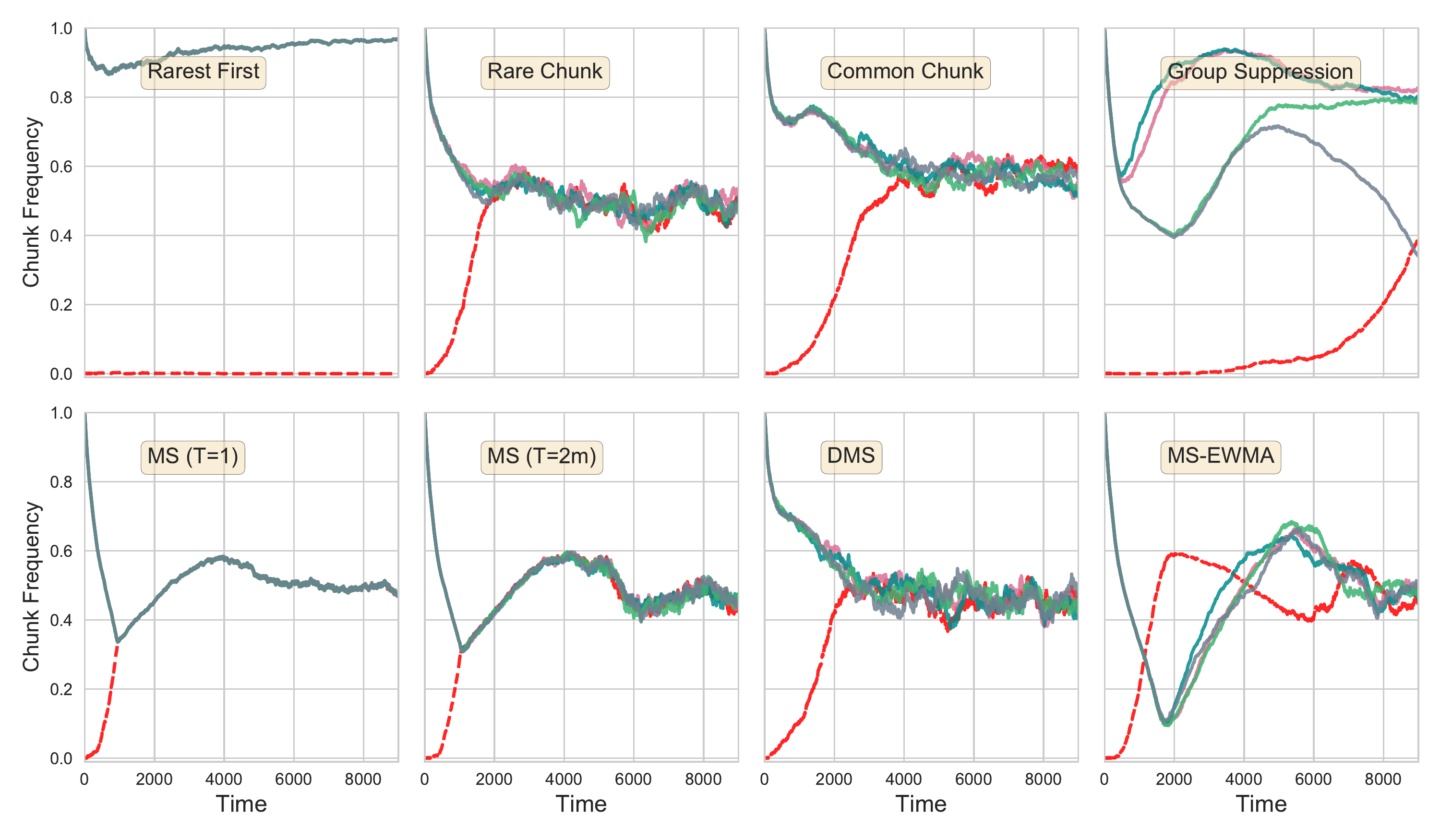}
    \caption{
    Chunk frequency evolution in a system with $m=5$ chunks under different policies when starting from the state of a ``missing-chunk'' (whose frequency is indicated by a red/dashed line).  Rarest-first is clearly unstable, since it cannot recover, whereas the other protocols manage to bring the chunk back into peer circulation and stabilize the system. }
  \label{fig:chunk-evolution}
\vspace{-0.2in}
\end{figure*}

\subsection{Chunk Frequency Evolution}
A stable chunk selection policy has to be robust to the one-club state.
In other words, a stable policy should be able to boost the frequency of a rare chunk.
To see how different policies handle the one-club situation, we start the system with 500 peers that have all the chunks except first chunk (i.e., all peers are part of the one-club).
In Figure~\ref{fig:chunk-evolution}, we plot the evolution of the chunk frequency for different policies under this initial condition.
We see that when using the rarest-first policy, the rare chunk remains rare and abundant chunks remain abundant, which is a clear sign of instability.
In all stabilizing policies, the rare chunk is made available by giving priority to that chunk in some way.
For instance, in case of mode-suppression ($T=1$), no other chunk will be transmitted until the frequency of the rare chunk is equal to the frequency of all other chunks.
Once this happens, the frequencies of the different chunks remain almost same, and hence we only see a thin spread across the frequencies.
Other policies also manage to bring the rare chunk back into circulation and the corresponding statistics become similar to all other chunks.
We also observe that the \emph{stabilization time} to increase the frequency of rare chunk to the same level as that of other chunk frequencies, is shorter for MS and DMS when compared to other algorithms.

\begin{figure*}[t]
  \centering
    \includegraphics[width=0.7\textwidth]{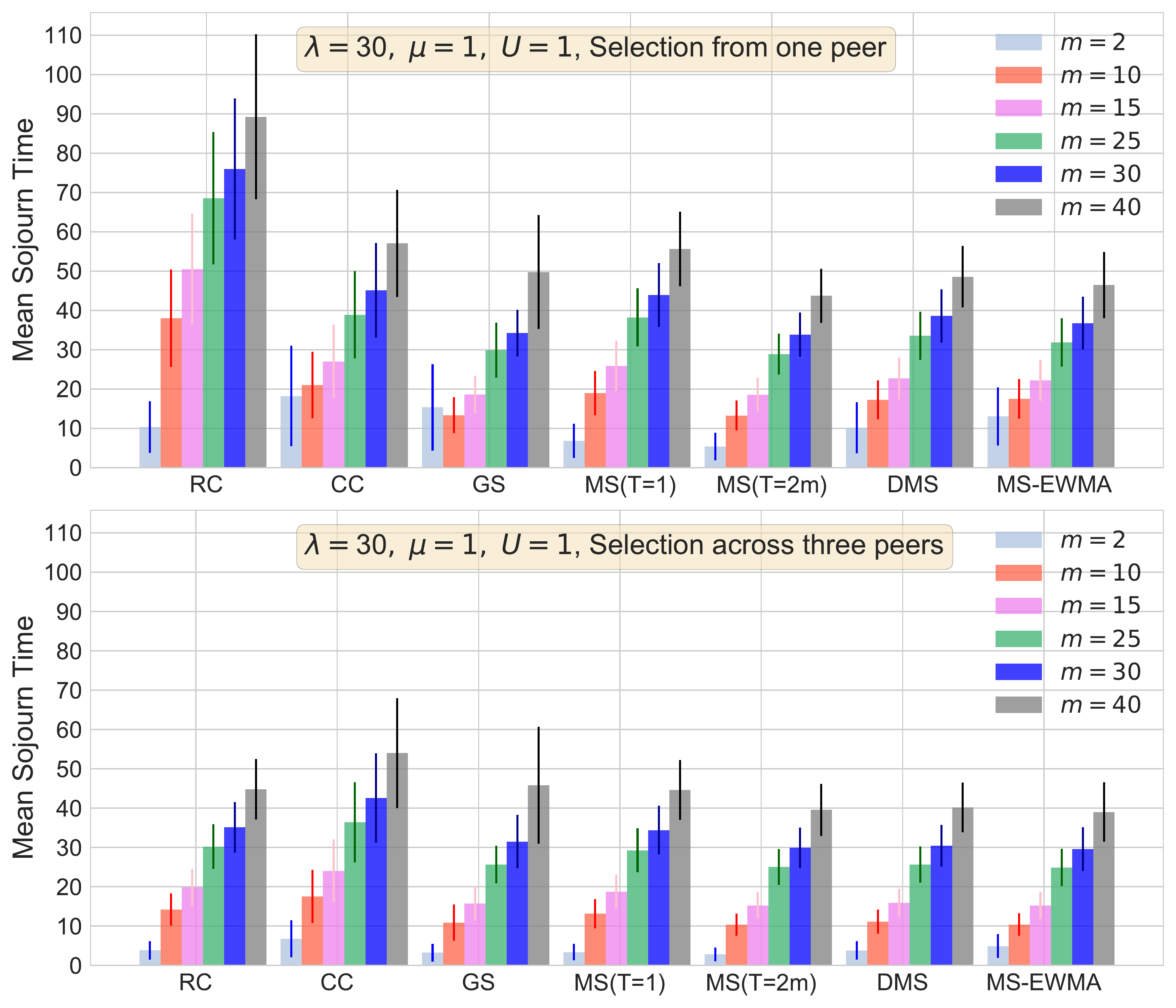}
    \caption{Stationary mean sojourn times of stable policies for different values of $m.$ The two regimes are to download a chunk from 1 peer, or to download one chunk from the chunk set of 3 peers.}
  \label{fig:sojourn-times}
\end{figure*}
\subsection{Sojourn times}
In a stable system, an important performance metric is the sojourn time of a peer, which is defined as the amount of time a peer spends in the system collecting all chunks before leaving.  For numerical illustration of sojourn time, we fix the peer arrival rate at $\lambda =30,$ and we calculate the mean stationary sojourn times of the peers under different policies, for different values of the number of file chunks $m.$  The stationary sojourn times are obtained by running the system for a long period of time and ignoring the first 2000 peers that left the system.    Our goal is to evaluate how effectively the algorithms use their  information on chunk statistics.

Our first result is on determining the value of threshold $T$ that minimizes the sojourn time under MS.  Intuitively, the threshold is a way of allowing ``noisy'' suppression of the mode.  It seems reasonable that as the number of chunks increases, the amount of noise permitted should also be allowed to increase in the interest of allowing more sharing to take place .  Thus, we numerically studied different values of $T$ that are increasing with the number of chunks $m,$ and found empirically that setting $T=2m$ appears to minimize the sojourn time under MS.

We also wish to study the effect of chunk diversity provided through the ability to choose a chunk from the set of chunks possessed by 1 versus 3 peers.  Thus, we have two versions of each algorithm that both use identical chunk statistics (obtained through sampling some or all peers as per the algorithm).  However, the first version  can obtain any one chunk from those possessed by 1 randomly selected peer, while the second can pick any one chunk from the set of chunks possessed by 3 randomly selected peers.

In Figure~\ref{fig:sojourn-times}, we present a  comparison of sojourn times across the different algorithms.  The increased sojourn times of RC and CC are visible, although increasing chunk diversity by sampling 3 peers improves RC considerably. GS has good performance, although the variability in sojourn time seen in the error bars (standard deviation) is high, particularly when $m$ is large.  The variants of MS all perform well, with the MS ($T=2m$), DMS and MS-EWMA all showing low sojourn times.   It is interesting to note that in the example, since the contact rate is 1, the best case sojourn time is equal to the number of chunks $m.$  We see that for the case of sampling 3 peers, the mode-suppression variants  MS($T=2m$), DMS and MS-EWMA attain a mean sojourn time that is very close to $m,$  indicating that they achieve a near-optimal tradeoff between suppression (to keep peers in the system)  and sharing (to enable peers to gather chunks).

\section{Conclusion}

In this work, we analyzed the scaling behavior of a P2P swarm with reference to its stability when subjected to an arbitrary arrival rate of peers.   It has been shown earlier that not all chunk sharing policies are stable in such a regime, and our goal was to design a simple and stable policy that yields low sojourn times.  Our main observation was that, contrary to the traditional approach of boosting the availability of rare chunks, preventing the spread of chunk(s) that are more frequent as compared to the lowest frequency chunks (where the maximum allowed threshold is a parameter of the algorithm) yields a simple and stable policy that we entitled mode-suppression (MS).  We analytically proved its stability, and showed that the sojourn time under this algorithm does not scale up with increasing demand (peer arrival rate).   We also described distributed versions of the policy that work on the same principle, but do not require global chunk frequency estimates.  Our results indicate that there is a delicate trade-off between sharing (i.e., uploading a useful chunk if at all possible) and suppression (i.e., trying to reduce chunk transfers to keep peers in the system so that they can help others).  We showed in numerical studies that MS with an appropriately selected threshold, as well as the heuristic distributed versions yield low (near-optimal) sojourn times.  An additional observation is that it appears that the chunk diversity provided by choosing a chunk from the set possessed by three randomly selected peers is sufficient for attaining this near-optimal performance.



\end{document}